\documentclass[11pt]{article}
\usepackage[margin=1in]{geometry}
\usepackage{amsmath,amssymb,amsthm,booktabs,array,microtype,mathrsfs,listings}
\usepackage[colorlinks=true,allcolors=blue]{hyperref}
\newtheorem{proposition}{Proposition}
\newtheorem{theorem}{Theorem}

\newcommand{\dd}{\mathrm d}
\newcommand{\HH}{\mathcal H}
\newcommand{\PP}{\mathcal P}
\newcommand{\al}{\alpha}
\title{Towards a Complete Kerr Soft-Hair Theory:\\Charge Integrability, Counterterm Equivalence,\\and Quantum State Counting}
\author{Sandi Setiawan\\\small Life Member, Clare Hall, University of Cambridge\\\small email: sandi.setiawan@cantab.net}
\date{October 7, 2026}
\begin{document}
\maketitle
\begin{abstract}
The horizon symmetry construction of Haco, Hawking, Perry and Strominger yields two Virasoro central charges $c_L=c_R=12J$ after a specified Wald--Zoupas correction. Turning that calculation into a black hole microstate theory requires separate control of nonlinear charge integrability, the allowed counterterm class, and quantum state multiplicities. We formulate exact local and global integrability tests, prove uniqueness of the counterterm coefficient within a one-parameter ansatz, and identify the quotient in which a physically meaningful uniqueness theorem must be stated. A coadjoint-orbit model supplies globally defined classical moment maps, while unitary highest-weight modules supply an explicit algebraic quantum realization for macroscopic positive angular momentum. Neither construction is proved equivalent to the gravitational phase space. A character calculation shows why a single Virasoro representation, despite having the desired central charge, does not provide the required large-central-charge entropy. Additional calculations exhibit a nonzero horizon-charge curl on the exact Kerr family, nonlinear integrable charges in a smooth isolated-horizon sector with zero central extension, and a regularity obstruction to identifying that sector with the HHPS modes. We derive a curvature improvement invisible to the constant mode, quantize a regulated area--boost corner polarization, and find a continuous kinematic spectrum with divergent Gibbs trace. Exact descendant counts and a Kerr spin threshold sharpen the state-counting limitations. Mixed temperature--mode derivatives, a moving-regulator cancellation, and a weighted-space trace test identify further requirements on a nonlinear phase-space construction. Explicit local modular CFT spectra with identical central charges yield different finite-temperature entropies, demonstrating that the available symmetry and quantum consistency data do not select the microscopic theory. We give a conditional sufficiency theorem and a concrete sequence of calculations needed to establish a gravitational completion. The paper supplies mathematical criteria and obstruction tests; it does not claim an unconditional first-principles derivation of Kerr microstates.
\end{abstract}

\section{Introduction and scope}
For a generic rotating Kerr black hole, Haco, Hawking, Perry and Strominger (HHPS) construct left and right horizon vector fields and evaluate a central extension using covariant surface charges \cite{hhps}. Their result is remarkable: after adding a geometric boundary correction, the central charges lose their dependence on the left and right temperatures and become $12J$. Inserting these values in the canonical Cardy formula yields the Bekenstein--Hawking entropy.

The paper nevertheless explicitly distinguishes removal of this obstruction from proof of the existence of fully integrable quantum charges. Its Section 5 leaves a comprehensive phase-space analysis to future work and does not establish uniqueness of the correction. These qualifications matter because a central extension is only one datum of a microscopic theory. It does not specify the phase space, its global topology, the quantum representation content, or the spectral density.

Our purpose is to address these gaps with statements that can be checked independently. We separate four levels: the supplied horizon integral calculation; rigorous consequences of that calculation; explicitly constructed auxiliary models; and unproved conditions needed to identify those models with gravity. Standard symplectic and representation-theoretic constructions are used as tools, not presented as new discoveries. The significance of the analysis is the separation of the tests needed for completion and the explicit demonstration of what the symmetry data fail to determine.

A fully first-principles completion cannot be obtained by declaring the missing conditions true. In particular, choosing a Hilbert space whose entropy already equals the area would merely insert the answer. The present work avoids that circular step.

We set $G=\hbar=c_{\rm light}=k_B=1$ unless stated otherwise and choose the azimuthal orientation so that $J>0$. Central charge is dimensionless; in physical units $c_L=c_R=12J_{\rm physical}/\hbar$. The symbol $J$ is not an energy in joules. We restrict attention to nonextremal rotating Kerr backgrounds. The limits $a\to0$ and $T_R\to0$ require separate coordinate and boundary analyses.

\section{What the horizon calculation establishes}
Let $a=J/M$ and
\begin{equation}
 r_\pm=M\pm\sqrt{M^2-a^2},\qquad
 T_L=\frac{r_++r_-}{4\pi a},\qquad
 T_R=\frac{r_+-r_-}{4\pi a}.
\end{equation}
The temperatures are dimensionless quantities conjugate to the conformal zero modes, not the Hawking temperature itself. HHPS compute a linearized charge with an Iyer--Wald contribution and a Wald--Zoupas correction \cite{iw,wz}:
\begin{equation}
 \al_\zeta[g;\delta g]=\delta Q_\zeta^{\rm IW}+\delta Q_\zeta^X,
 \qquad \delta Q_\zeta^X=\frac1{16\pi}\int_S\iota_\zeta(*X).
\end{equation}
Here $S$ is approached as a regulated bifurcation surface, $X$ is a spacetime one-form linear in the perturbation, and $\al_\zeta$ is a one-form on the space of solutions. Their candidate is
\begin{equation}\label{candidate}
 X_a^{(0)}=2h_{ab}\Omega^b,\qquad
 \Omega_a=q_a{}^c n_b\nabla_c l^b,\qquad l\cdot n=-1.
\end{equation}
We retain the HHPS metric-variation convention for this formula; changing from inverse-metric to covariant-metric perturbations requires the corresponding sign conversion. $\Omega_a$ is the normal-bundle connection determined by the chosen null frame.

Using the surface integrals supplied in HHPS, the cubic central coefficients are
\begin{align}
 K_R^{\rm IW}&=\frac{2JT_R}{T_L+T_R}m^3\delta_{m+n,0},
 &K_R^X&=\frac{J(T_L-T_R)}{T_L+T_R}m^3\delta_{m+n,0},\\
 K_L^{\rm IW}&=\frac{2JT_L}{T_L+T_R}m^3\delta_{m+n,0},
 &K_L^X&=\frac{J(T_R-T_L)}{T_L+T_R}m^3\delta_{m+n,0}.
\end{align}
Consequently $K_A=Jm^3\delta_{m+n,0}$ for $A=L,R$, and comparison with $c_A m^3/12$ gives $c_A=12J$. We check the combination of the quoted terms; we do not claim a new independent tensor-level evaluation of their integrands.

\subsection{A restricted coefficient uniqueness result}
Consider only the ansatz $X^{(\lambda)}=\lambda X^{(0)}$, with constant $\lambda$, and retain the supplied integral formulas. Write $x=T_R/(T_L+T_R)$. Then
\begin{equation}\label{lambda}
 \frac{K_R^{(\lambda)}}{Jm^3\delta_{m+n,0}}
 =\lambda+2(1-\lambda)x,
 \qquad
 \frac{K_L^{(\lambda)}}{Jm^3\delta_{m+n,0}}
 =2-\lambda+2(\lambda-1)x.
\end{equation}
The ratios in this equation denote the coefficients of the indicated mode structure, including when the Kronecker delta vanishes.
\begin{proposition}
On a fixed-$J$ Kerr family over an open interval of nonextremal masses, constancy of either cubic coefficient in \eqref{lambda} forces $\lambda=1$. Both central charges then equal $12J$.
\end{proposition}
\begin{proof}
At fixed positive $J$, $a=J/M$ and
\begin{equation}
 x=\frac{\sqrt{M^2-a^2}}{M+\sqrt{M^2-a^2}}
\end{equation}
varies with $M$. The coefficients of $x$ must therefore vanish. Each condition gives $\lambda=1$.
\end{proof}
This proves uniqueness of a coefficient in a specified one-dimensional ansatz. It does not prove uniqueness among all local covariant counterterms, nor does a test at only one background establish this result. The sum $K_L^{(\lambda)}+K_R^{(\lambda)}=2Jm^3\delta_{m+n,0}$ is independent of $\lambda$, so testing only the sum cannot select the correction.

\section{Full integrability is a field-space problem}
\subsection{The exact criterion}
Let $\PP_J$ be a proposed reduced space of on-shell gravitational configurations at fixed angular momentum. The gauge quotient, boundary conditions, function spaces, and regulator must be specified before this notation denotes a defined mathematical object. A charge exists locally precisely when its variation one-form is locally exact:
\begin{equation}\label{int}
 \al_\zeta=\dd_{\PP}Q_\zeta.
\end{equation}
For tangent vector fields $\delta_1,\delta_2$ on $\PP_J$, define the obstruction
\begin{equation}\label{curl}
 \mathcal I_\zeta(\delta_1,\delta_2)
 =\delta_1[\al_\zeta(\delta_2)]-
 \delta_2[\al_\zeta(\delta_1)]-
 \al_\zeta([\delta_1,\delta_2]).
\end{equation}
This is exactly $(\dd_{\PP}\al_\zeta)(\delta_1,\delta_2)$.
\begin{proposition}\label{integrability}
On a contractible smooth chart where the Poincare lemma applies, $\mathcal I_\zeta=0$ for all tangent pairs is equivalent to local charge integrability. On a connected phase space, global integrability additionally requires all periods $\oint_\gamma\al_\zeta$ to vanish.
\end{proposition}
\begin{proof}
Exactness implies closure and zero periods. Conversely, local closure gives local primitives. If all periods vanish, integrating $\al_\zeta$ from a base point is path independent and defines a global primitive.
\end{proof}
For infinite-dimensional phase spaces, the smooth structure and applicable local exactness theorem must be supplied; the finite-dimensional Poincare lemma cannot simply be invoked without functional-analytic assumptions.

Constancy of a central coefficient examines particular contractions involving symmetry variations. Equation \eqref{curl} tests arbitrary tangent pairs. For example, on an auxiliary phase space with coordinates $(x,y,u,v)$, adding $u\,\dd v$ to a charge one-form leaves every evaluation on symmetry vectors tangent only to $(x,y)$ unchanged, while adding $\dd u\wedge\dd v$ to its curl. This elementary example is not asserted to be a gravitational counterterm. It proves the logical point: data restricted to symmetry directions do not determine closure in transverse directions.

\subsection{Flux and boundary corrections}
For field-independent $\zeta$, a fixed integration surface, and the conventional form $k_\zeta^{\rm IW}=\delta\mathbf Q_\zeta-\iota_\zeta\theta$, the Iyer--Wald obstruction is
\begin{equation}\label{flux}
 (\dd_{\PP}\al_\zeta^{\rm IW})(\delta_1,\delta_2)
 =-\int_S\iota_\zeta\omega(g;\delta_1g,\delta_2g),
\end{equation}
where $\omega=\delta_1\theta[\delta_2g]-\delta_2\theta[\delta_1g]$ for commuting variations. Thus a correction must satisfy
\begin{equation}\label{cancel}
 \dd_{\PP}\al_\zeta^X
 =\int_S\iota_\zeta\omega
\end{equation}
as a two-form identity on the allowed phase space. This is stronger than canceling the temperature dependence of $K_{m,n}$. In radiative settings the Wald--Zoupas prescription may define a flux-balanced charge rather than an ordinary Hamiltonian on a closed subsystem \cite{wz}. A completion must say which interpretation is intended and account for any boundary flux.

If the generators or surfaces depend on $g$, extra terms enter. Equation \eqref{curl}, applied to the total charge one-form, remains the safe definition. One either freezes the mode frame as boundary data or uses the adjusted bracket
\begin{equation}
 [\zeta,\eta]_*=[\zeta,\eta]-\delta_\zeta\eta+\delta_\eta\zeta
\end{equation}
with a consistently defined action convention. One must include the variation of $\zeta$ in the total differentiation and, in the usual surface-charge prescription, subtract the Noether charge associated with $\delta\zeta$. The parameter dependence of the HHPS temperatures makes this issue material when mass is allowed to vary.

\subsection{A regulated gravitational calculation}
At a nonzero cutoff $\epsilon$, define $\al_{\zeta,\epsilon}$ on a stretched surface $S_\epsilon$ using a fixed identification with an abstract corner. A successful proof requires a limit
\begin{equation}
 \al_\zeta=\lim_{\epsilon\to0}\al_{\zeta,\epsilon},\qquad
 \dd_{\PP}\al_\zeta=\lim_{\epsilon\to0}\dd_{\PP}\al_{\zeta,\epsilon},
\end{equation}
with estimates that justify interchange of variation and limit. Finiteness of one central integral does not justify either interchange. The oscillatory modes and poles used by HHPS must remain admissible under the proposed nonlinear boundary conditions.

The perturbation pairs must include symmetry--symmetry, symmetry--regular, regular--regular, and stationary-parameter--mode variations whenever these are allowed. One must establish the linearized constraints for each pair, closure of the boundary conditions under finite transformations, and the role of horizon shear or news. If radiative modes are excluded, this exclusion is a physical restriction, not a consequence of the central calculation.

The stationary Kerr family alone cannot test full integrability. At fixed $J$ it is locally one-dimensional, so every smooth two-form vanishes on it. Closure there is automatic and says nothing about a phase space containing independent horizon excitations. Conversely, checking only pure diffeomorphisms may describe a symmetry orbit while leaving transverse gravitational perturbations untested.

\section{What counterterm uniqueness can mean}
\subsection{A charge map and its kernel}
Fix an admissible class $\mathscr X$ of local covariant, perturbation-linear one-forms with a stated derivative bound, null-frame transformation law, and boundary behavior. Define
\begin{equation}
 \mathscr C(X)_\zeta(\delta g)
 =\frac1{16\pi}\lim_{\epsilon\to0}\int_{S_\epsilon}
 \iota_\zeta(*X[g;\delta g]).
\end{equation}
A counterterm lies in $\ker\mathscr C$ if it contributes zero for every allowed perturbation and generator. The kernel depends on the boundary conditions and regulator. It cannot be inferred from evaluations on $\delta g=\mathcal L_{\zeta_n}g$ alone.

Suppose $X$ and $X'$ both solve the same integrability equation with fixed Iyer--Wald data. Their difference obeys
\begin{equation}\label{diff}
 \dd_{\PP}\mathscr C(X'-X)_\zeta=0.
\end{equation}
Locally, $\mathscr C(X'-X)_\zeta=\dd_{\PP}F_\zeta$. This identity classifies the remaining charge improvement; it does not force $F_\zeta=0$. The linearized map should be analyzed before taking a quotient by invisible or conventional terms.

\begin{proposition}
If two sets of globally integrable charges generate the same vector fields with respect to the same nondegenerate reduced symplectic form, their differences are constant on each connected component. Fixing their values at a reference configuration then makes the charges identical.
\end{proposition}
\begin{proof}
With the convention $\dd Q_\zeta=\iota_{V_\zeta}\Omega_{\rm red}$, both differentials are identical. Therefore $\dd(Q'_\zeta-Q_\zeta)=0$.
\end{proof}
This is uniqueness of normalized Hamiltonians for an already specified symplectic action. It is not uniqueness of a spacetime representative $X_a$, and does not apply if the correction also changes the symplectic structure or the boundary action.

\subsection{Frame dependence and a restricted classification}
The rotation one-form transforms under a null-frame boost as
\begin{equation}
 l\mapsto e^\sigma l,\qquad n\mapsto e^{-\sigma}n,
 \qquad \Omega_a\mapsto\Omega_a-D_a\sigma
\end{equation}
with the definition in \eqref{candidate}. Hence $X^{(0)}$ changes by $-2h_{ab}D^b\sigma$. Its integrated contribution must be checked; it is not automatically an exact charge variation for arbitrary perturbations. A fixed boost gauge or an invariant corner formalism is needed.

Within a finite derivative class, the classification procedure is to enumerate independent tensor structures from $q_{ab},l^a,n^a,\Omega_a$, extrinsic curvature, and the perturbation; impose the boost rule; calculate the complete curl; and determine the kernel of the resulting charge map. Examples that may enter a broader ansatz include terms proportional to $h\Omega_a$, projected components of $h_{ab}\Omega^b$, and derivatives of $h$. They are possible structures to test, not verified admissible solutions. Curvature-dependent coefficients enlarge the class further unless restricted independently.

A defensible uniqueness statement would say that every admissible solution differs from $X^{(0)}$ only by charge-invisible terms and improvements that preserve the same normalized Hamiltonian action. This paper proves the coefficient result of Section 2, but not this complete local classification. Requiring $c=12J$ as an extra selection rule would fix the desired answer rather than derive it.

\section{An exactly integrable classical symmetry model}
A controlled model of the symmetry sector is obtained from Virasoro coadjoint orbits \cite{witten}. Let $\widehat{\mathfrak g}$ be a centrally extended algebra and $b$ an element of its dual. On a coadjoint orbit $\mathcal O_b$, define the Kirillov--Kostant--Souriau form by
\begin{equation}
 \Omega_b(V_u,V_v)=\langle b,[u,v]\rangle,
 \qquad V_u=\operatorname{ad}_u^*b.
\end{equation}
We use $\langle\operatorname{ad}_u^*b,v\rangle=-\langle b,[u,v]\rangle$, so for $\mu_u(b)=\langle b,u\rangle$ one has
\begin{equation}
 \dd\mu_u(V_v)=\langle b,[u,v]\rangle
 =\Omega_b(V_u,V_v).
\end{equation}
The charges are globally specified functions, the symplectic form is closed by the Jacobi identity, and degeneracies associated with the stabilizer are removed in defining the orbit. The centrally extended algebra is built into this construction.

Two such orbits at central values $12J$ yield an explicit integrable classical $\mathrm{Vir}_L\oplus\mathrm{Vir}_R$ model. The orbit labels, stabilizers, and Hamiltonian remain choices. The construction demonstrates that the desired algebra is compatible with an integrable symmetry phase space; it does not derive its central value from gravity.

To promote this model to a completion of HHPS one must construct a map
\begin{equation}\label{map}
 \Phi:\PP_J^{\rm grav}\longrightarrow\mathcal O_L\times\mathcal O_R
\end{equation}
and prove, for the sector claimed to be described,
\begin{equation}
 \Phi^*(\Omega_L+\Omega_R)=\Omega_{\rm grav,red},\qquad
 \al_{\zeta_n}^{\rm grav}=\dd_{\PP}(\mu_n\circ\Phi).
\end{equation}
One must also prove that the map respects the constraints and mode action and state whether it is onto, one-to-one after reduction, or only an embedding of a subsector. Defining the gravitational symplectic form to be the pullback would impose the conclusion. Additional bulk degrees of freedom may require extra symplectic factors or a fibration rather than two orbits.

\section{An explicit quantum algebra and its limitations}
\subsection{Highest-weight realization}
Use the standard quantum Virasoro convention
\begin{equation}\label{vir}
 [L_m,L_n]=(m-n)L_{m+n}+\frac c{12}m(m^2-1)\delta_{m+n,0},
 \qquad L_n^\dagger=L_{-n}.
\end{equation}
The classical cubic convention is related by a zero-mode shift: if $\widetilde L_n=L_n-(c/24)\delta_{n,0}$, the central term is $cm^3/12$. Factors of $i$ in a real classical basis are conventional and must be restored consistently when matching Poisson brackets to commutators.

For $c>1$ and $h>0$, unitary positive-energy highest-weight representations exist \cite{fqs}. Start from a vector $|h\rangle$ satisfying
\begin{equation}
 L_0|h\rangle=h|h\rangle,\qquad L_n|h\rangle=0\quad(n>0).
\end{equation}
Act with the negative modes, use \eqref{vir} to define the Gram form, quotient null vectors if present, and complete the positive inner-product space. For example,
\begin{equation}
 \|L_{-n}|h\rangle\|^2=2nh+\frac c{12}n(n^2-1).
\end{equation}
These individual norms illustrate positivity but are not by themselves a proof of positivity of all Gram matrices; the full construction uses standard unitary representation theory. The algebra acts on the invariant dense domain of finite-level descendants. Claims about global exponentiation and self-adjoint generators require their additional analytic hypotheses.

For $J>1/12$, one may therefore choose $c=12J>1$ and define
\begin{equation}\label{hilb}
 \HH_{\rm alg}=\HH_{c,h_L}\otimes\HH_{c,h_R}.
\end{equation}
The two algebras commute and have the requested central values. This is an explicit algebraic realization, not a quantization map from the gravitational charge functions. It leaves $h_L,h_R$ undetermined and provides no proof that the classical central value has no quantum correction. For smaller $J$, unitary representation restrictions must be reconsidered. Macroscopic Kerr black holes are in the large-$J$ regime in Planck units.

\subsection{Why a module is not a microstate theory}
For a generic nondegenerate module, the character is
\begin{equation}\label{character}
 \chi_{c,h}(q)=q^{h-c/24}\prod_{n=1}^{\infty}(1-q^n)^{-1}.
\end{equation}
At descendant level $N$, the number of states is the partition number $p(N)$, so
\begin{equation}
 \log p(N)=2\pi\sqrt{\frac N6}+O(\log N).
\end{equation}
The large-$N$ descendant coefficient does not grow with $c$. By contrast, the vacuum-dominated Cardy density of a suitable full conformal theory has an exponent involving its effective central charge. At levels of order $c$, the entropy of one generic module is only of order $\sqrt c$, while the black hole entropy sought in the semiclassical regime is of order $c$. Extra primary content is essential.

The same point can be checked in the canonical ensemble. With $q=e^{-\beta}$, modular asymptotics of the Euler product give
\begin{equation}
 \log\chi_{c,h}(e^{-\beta})
 =\frac{\pi^2}{6\beta}+\frac12\log\frac{\beta}{2\pi}
 -\beta\left(h-\frac{c-1}{24}\right)+O(e^{-4\pi^2/\beta}).
\end{equation}
Consequently,
\begin{equation}
 S_{\rm module}=(1-\beta\partial_\beta)\log\chi
 =\frac{\pi^2}{3\beta}+\frac12\log\frac{\beta}{2\pi}-\frac12
 +O(\beta^{-1}e^{-4\pi^2/\beta}).
\end{equation}
The vacuum-energy shift cancels from the entropy. Increasing the algebraic central charge alone does not replace the coefficient $1$ of the oscillator product by $c$.

\begin{proposition}
The values $c_L=c_R=12J$ and existence of unitary representations do not determine the black hole state density.
\end{proposition}
\begin{proof}
A single tensor product \eqref{hilb} and a finite direct sum of several copies obey identical Virasoro commutators, but have different degeneracies. Varying the highest weights changes the spectrum without changing $c$. More generally, the algebra does not specify the set of primary modules or their multiplicities. The character calculation shows that descendants alone do not automatically supply the large-$c$ Cardy density.
\end{proof}
This is an underdetermination result for the symmetry data. It does not rule out a gravitational theory with a uniquely determined and sufficiently rich spectrum. Such a theory must derive that spectrum rather than choose it to match the area.

\section{From quantum symmetry to area entropy}
A prospective full quantum space has a decomposition of the form
\begin{equation}
 \HH_J=\bigoplus_{h_L,h_R}N_{h_L,h_R}
 \HH_{c,h_L}\otimes\HH_{c,h_R},
 \qquad N_{h_L,h_R}\in\mathbb Z_{\ge0},
\end{equation}
or an appropriate direct integral for a continuous spectrum. Virasoro symmetry does not determine the multiplicities or measure. A local CFT completion requires compatibility with operator products, locality, level matching where applicable, and modular consistency. These are independent constraints.

If a suitable modular completion has a normalizable vacuum and vacuum dominance in the relevant transformed channel, then the high-temperature canonical formula \cite{cardy} is
\begin{equation}\label{cardy}
 S=\frac{\pi^2}{3}(c_LT_L+c_RT_R)+\text{subleading terms}.
\end{equation}
If the lowest weights are nonzero or the spectrum is noncompact, the effective central charge and spectral measure require separate attention. Large $c$ alone does not guarantee applicability at all temperatures; sparse-spectrum hypotheses extend universality only under defined conditions \cite{hks}. In particular the near-extremal right-moving temperature is small, so a two-sector high-temperature argument is not uniformly valid for all Kerr parameters.

Assuming \eqref{cardy} with the central values above,
\begin{align}
 S&=4\pi^2J(T_L+T_R)
 =4\pi^2J\frac{r_+}{2\pi a}
 =2\pi Mr_+,\\
 A&=4\pi(r_+^2+a^2)=8\pi Mr_+,
 \qquad S=\frac A4.
\end{align}
The last equality uses the horizon equation $r_+^2-2Mr_++a^2=0$. This is an exact algebraic check of the leading formula under its statistical assumptions, not an exact microscopic result including finite-$J$ corrections. With physical constants, $S_{\rm BH}/k_B=A/(4\ell_P^2)$.

A useful numerical illustration is $M=2$, $a=1$, $J=2$. Then $r_+=2+\sqrt3$, $c_L=c_R=24$, $T_L=1/\pi$, and $T_R=\sqrt3/(2\pi)$. The conditional Cardy expression gives $4\pi(2+\sqrt3)$, equal to $A/4$. These temperatures are not parametrically high; numerical agreement is an algebraic identity and cannot validate the ensemble assumptions.

\section{Gravitational gluing and physical states}
Boundary modes can be physical for a subregion while full diffeomorphisms remain constraints after gluing. An extended phase space with corner embedding and normal-frame variables is a concrete starting point \cite{df}; it does not alone identify the HHPS Virasoro sector or its entropy.

Suppose an extended quantum description permits sector decompositions
\begin{equation}
 \HH_{\rm in}^{\rm ext}=\bigoplus_\lambda M_\lambda^{\rm in}\otimes V_\lambda,
 \qquad
 \HH_{\rm out}^{\rm ext}=\bigoplus_\lambda M_\lambda^{\rm out}\otimes V_\lambda^*.
\end{equation}
Here $M_\lambda$ denotes multiplicity spaces and $V_\lambda^*$ is the appropriate algebraic dual, with opposite central action under the dual representation. When gluing is mathematically defined, invariant contraction removes the representation factor but does not in general remove its multiplicity spaces. Even for finite-dimensional unitary groups, singlet constraints leave factors $M_\lambda^{\rm in}\otimes M_\lambda^{\rm out}$.

For Virasoro representations, formal invariant pairings may be distributional rather than normalizable. For an infinite orthonormal basis, the schematic singlet $\sum_i|i\rangle\otimes|i\rangle^*$ has divergent norm. A regulator and a rigging or gluing map are therefore required. If interior and exterior constraints are incorrectly assigned same-sign central terms, their algebra is anomalous. Consistent orientation and dualization, or compensating sectors, must be established.

Thus gravitational constraints do not by themselves prove an isometric identification of the entire physical Hilbert space with an exterior Hilbert space. They enforce matching conditions. A microscopic completion must calculate the resulting physical inner product, account for multiplicities, and define which entropy is counted. This also prevents counting gauge copies as independent microstates.

\section{An explicit integrability test on the Kerr solution family}\label{kerrcurl}
The general criterion can be evaluated on an actual two-dimensional family of vacuum solutions. Work in geometrized Planck units and define
\begin{equation}
 D=\sqrt{M^4-J^2},\qquad M>0,\quad 0<J<M^2.
\end{equation}
The Kerr area, surface gravity, and horizon angular velocity become
\begin{equation}\label{kerrfunctions}
 A=8\pi(M^2+D),\qquad
 \kappa=\frac{D}{2M(M^2+D)},\qquad
 \Omega_H=\frac{J}{2M(M^2+D)}.
\end{equation}
Their derivatives give
\begin{equation}
 \partial_M A=16\pi M\left(1+\frac{M^2}{D}\right),\qquad
 \partial_J A=-\frac{8\pi J}{D}.
\end{equation}
Multiplying by $\kappa/(8\pi)$ verifies the stationary first law directly:
\begin{equation}\label{firstlawform}
 \frac{\kappa}{8\pi}\dd A=\dd M-\Omega_H\dd J.
\end{equation}
The relevant horizon generator is $\chi[g]=\partial_t+\Omega_H[g]\partial_\phi$. With the field-dependent surface-charge prescription that excludes the spurious $\mathbf Q_{\delta\chi}$ contribution, its variation is the one-form on the right of \eqref{firstlawform}. It is not the total derivative of $M-\Omega_HJ$, which would include $-J\dd\Omega_H$.

\subsection{A nonzero curl on a genuine solution space}
Let $\al_\chi=\dd M-\Omega_H\dd J$. Differentiation of \eqref{kerrfunctions} yields
\begin{equation}
 \partial_M\Omega_H=-\Omega_H\left(\frac1M+\frac{2M}{D}\right).
\end{equation}
Consequently
\begin{equation}\label{kerrnonclosed}
 \boxed{\dd\al_\chi=
 \Omega_H\left(\frac1M+\frac{2M}{D}\right)\dd M\wedge\dd J.}
\end{equation}
For $J>0$ this is strictly nonzero throughout the nonextremal domain. Thus the family-dependent horizon generator does not define a single integrable charge on this entire Kerr family in this normalization. The fixed asymptotic generators $\partial_t$ and $\partial_\phi$ still have integrable mass and angular-momentum charges. The obstruction concerns their particular field-dependent combination; it is not a failure of the first law or of Kerr mechanics.

Restriction to fixed $J$ removes the obstruction and gives $\al_\chi=\dd M$. Restriction to a one-dimensional curve also removes any two-form test by dimensionality. Neither restriction establishes integrability of the nonzero HHPS modes on the larger phase space.

\subsection{A finite-loop check}
Consider the positively oriented rectangle $[2,2+e]\times[2,2+e]$ in $(M,J)$ space. Since the horizontal segments give canceling integrals of $\dd M$, the exact circulation is
\begin{equation}
 \oint_{\partial R_e}\al_\chi=
 \int_{2}^{2+e}[\Omega_H(2,J)-\Omega_H(2+e,J)]\dd J.
\end{equation}
Stokes' theorem predicts its ratio to $e^2$ tends to the curl coefficient at $(2,2)$, namely $0.1108439182$. Composite Simpson quadrature gives:
\begin{center}
\begin{tabular}{rrr}
\toprule
$e$ & $\oint\al_\chi$ & $e^{-2}\oint\al_\chi$\\
\midrule
0.1000 & $1.023087745\times10^{-3}$ & 0.1023087745\\
0.0500 & $2.659333740\times10^{-4}$ & 0.1063733496\\
0.0250 & $6.784642463\times10^{-5}$ & 0.1085542794\\
0.0125 & $1.713828434\times10^{-5}$ & 0.1096850198\\
\bottomrule
\end{tabular}
\end{center}
The circulation is not a topological period of a closed form; it is an explicit local failure of closure. This calculation tests a gravitational zero-mode combination, not the full HHPS charge integrands.

\subsection{An integrating factor and its meaning}
Since $\kappa>0$ in this domain,
\begin{equation}\label{integratingfactor}
 \beta=\frac{\al_\chi}{\kappa}=\frac{\dd A}{8\pi},
 \qquad \dd\beta=0,\qquad B=\frac A{8\pi}+B_0.
\end{equation}
The area is a globally defined function, so the periods of $\beta$ vanish. Rescaling the generator to $\chi/\kappa$ requires consistent treatment of its field dependence. Algebraically, its adjusted charge variation is the form in \eqref{integratingfactor}. Multiplication by $2\pi$ gives the entropy variation $\dd S=\dd A/4$. This establishes an integrable area-related zero mode and makes the normalization dependence explicit. It does not construct a state degeneracy or solve the nonzero-mode problem.

\section{A smooth isolated-horizon sector and its charge algebra}\label{smoothsector}
We now calculate in a regular sector different from the oscillatory HHPS one. Near a null surface, use Gaussian null coordinates with
\begin{equation}\label{gnc}
 \dd s^2=-2\kappa\rho\dd v^2+2\dd v\dd\rho
 +2\rho\vartheta_A\dd v\dd x^A
 +(q_{AB}+\rho\lambda_{AB})\dd x^A\dd x^B+O(\rho^2).
\end{equation}
Let $a=\sqrt{\det q_{AB}}$ denote the area density. In this section it is not the Kerr rotation parameter. Smooth parameters are a scalar $T(x)$ and a vector field $Y^A(x)$ on a closed horizon section. Leading residual generators follow by preserving $g_{\rho\rho}=g_{\rho A}=0$ and $g_{v\rho}=1$:
\begin{equation}
 \xi^v=T,\qquad \xi^A=Y^A-\rho q^{AB}\partial_BT+O(\rho^2),
 \qquad \xi^\rho=\tfrac12\rho^2\vartheta^A\partial_AT+O(\rho^3).
\end{equation}
Choose nonexpanding data with $\partial_vq_{AB}=0$, $\partial_v\vartheta_A=0$, and constant $\kappa$ fixed across the sector. Existing gravitational surface-charge results give \cite{donnay}
\begin{equation}\label{smoothalpha}
 \al_{T,Y}=\frac1{16\pi G}\int_S
 [2\kappa T\,\delta a-Y^A\delta(a\vartheta_A)]\dd^2x.
\end{equation}
Here the leading symmetry parameters are held fixed under variation, while the subleading field dependence of the generators is handled with adjusted brackets. These conditions define a restricted boundary sector; remaining field equations must hold for any proposed bulk extension. We do not assert that arbitrary choices of $q_{AB}$ and $\vartheta_A$ extend to globally asymptotically flat Kerr exteriors.

\subsection{Nonlinear exactness within the stated sector}
Equation \eqref{smoothalpha} immediately integrates to
\begin{equation}\label{smoothQ}
 Q_{T,Y}=\frac1{16\pi G}\int_S a(2\kappa T-Y^A\vartheta_A)\dd^2x+Q^0_{T,Y}.
\end{equation}
No linear expansion of $a$ or $\vartheta_A$ is used: the expression is nonlinear in the boundary data. For commuting variations, the antisymmetrized second variation vanishes because $\delta_1\delta_2a=\delta_2\delta_1a$ and the same identity holds for $a\vartheta_A$. Thus
\begin{equation}
 \mathcal I_{T,Y}(\delta_1,\delta_2)=0.
\end{equation}
The explicit primitive proves zero periods on every connected chart where the smooth boundary fields and the reference constants are single valued. This is an integrability result for the specified charge one-form on allowed variations, without an assertion that the full gravitational reduced symplectic form has been reconstructed.

For $Y=0$ and $T=1$, $Q=\kappa A/(8\pi G)$ at fixed $\kappa$. Dividing this charge by the Hawking temperature $\kappa/(2\pi)$ gives $A/(4G)$. This classical zero-mode identity does not count quantum states.

\subsection{Direct calculation of the central term}
At the boundary, the transformations read
\begin{equation}\label{boundarytransform}
 \delta_{T,Y}a=\mathcal L_Ya,\qquad
 \delta_{T,Y}\vartheta_A=\mathcal L_Y\vartheta_A-2\kappa\partial_AT.
\end{equation}
Let $P_A=a\vartheta_A$ be a covector density and fix $Q^0=0$. The adjusted parameter bracket is
\begin{equation}
 T_{12}=Y_1^A\partial_AT_2-Y_2^A\partial_AT_1,
 \qquad Y_{12}=[Y_1,Y_2].
\end{equation}
For smooth fields on a closed surface, integration by parts gives
\begin{align}
 \delta_1Q_2&=\frac1{16\pi G}\int_S\big[
 2\kappa T_2\mathcal L_{Y_1}a-Y_2^A\mathcal L_{Y_1}P_A
 +2\kappa aY_2^A\partial_AT_1\big]\dd^2x\\
 &=\frac1{16\pi G}\int_S\big[-2\kappa aT_{12}+P_AY_{12}^A\big]\dd^2x
 =-Q_{12}.
\end{align}
With the bracket convention $\{Q_1,Q_2\}=-\delta_1Q_2$, we find
\begin{equation}\label{smoothalgebra}
 \boxed{\{Q_1,Q_2\}=Q_{12},\qquad K_{12}=0.}
\end{equation}
A different constant reference charge changes the algebra by a coboundary, not a nontrivial cubic Virasoro cocycle. Smoothness and the absence of puncture boundaries are essential to the integration by parts.

This result is a useful negative test: an integrable smooth horizon sector does not automatically reproduce the HHPS central extension. Its charge algebra is a semidirect product of surface diffeomorphisms with scalar translations. On a sphere, globally smooth conformal vector fields form the finite-dimensional conformal algebra; unrestricted local Laurent modes need singularities or punctures. No two nontrivial Virasoro central charges have been derived in this smooth closed-surface calculation.

\section{A precise regularity obstacle to identifying the sectors}\label{regularity}
The HHPS right-moving generator is
\begin{equation}
 \zeta_n=\epsilon_n\partial_++\tfrac12y\epsilon_n'\partial_y,
 \qquad \epsilon_n=\alpha w^+\exp\left(\frac{in}{\alpha}\log w^+\right),
 \qquad \alpha=2\pi T_R>0.
\end{equation}
For $w^+>0$, direct differentiation gives
\begin{align}
 \epsilon_n'&=(\alpha+in)\exp\left(\frac{in}{\alpha}\log w^+\right),\\
 \epsilon_n''&=\frac{in(\alpha+in)}{\alpha w^+}
 \exp\left(\frac{in}{\alpha}\log w^+\right).
\end{align}
For $n\ne0$, $\epsilon_n'$ has no limit as $w^+\to0^+$ and $\epsilon_n''$ has magnitude proportional to $1/w^+$. In a horizon chart at finite $y>0$, the component $\zeta_n^y=y\epsilon_n'/2$ is already oscillatory without a limiting value. The nonzero modes therefore are not smooth vector fields extending to that bifurcation surface in this chart. The same reasoning applies to the left sector.

A smooth Gaussian-null symmetry orbit cannot simply be identified with these modes by matching its charge labels. A proposed map must specify a weaker space of fields, a regulated surface prescription, or an excised boundary. It must then show that the symplectic form and nonlinear transformations are well defined there. These derivatives explain why the regulator is structural, not merely a way of evaluating an otherwise regular integral. They do not prove that a controlled weak or regulated HHPS phase space is impossible.

An elementary integral illustrates the issue. For one multiplicative azimuthal interval,
\begin{equation}
 \int_\epsilon^{\epsilon e^{2\pi\alpha}}\frac{\dd w}{w}
 =2\pi\alpha
\end{equation}
is finite and independent of $\epsilon$, although its integrand has a pole. For $n\ne0$,
\begin{equation}
 \int_\epsilon^{\epsilon e^{2\pi\alpha}}
 w^{in/\alpha}\frac{\dd w}{w}
 =\frac{\alpha}{in}\epsilon^{in/\alpha}(e^{2\pi in}-1)=0.
\end{equation}
The mode-projected integral exists while pointwise boundary data do not. This distinction is consistent with a finite central coefficient, but it supplies no dominated-convergence bound for arbitrary variations or a proof of nonlinear closure.

\section{An explicit local charge-improvement ambiguity}\label{improvement}
Integrability alone admits improvements that can evade zero-mode tests. Let $q_{AB}$ be a smooth metric on a closed two-sphere and $f$ a fixed scalar smearing function. Introduce a local functional
\begin{equation}\label{curvatureimprovement}
 F_f[q]=\lambda\int_S f R[q]\sqrt q\dd^2x.
\end{equation}
Adding $\dd_{\PP}F_f$ to a charge one-form preserves closure. This is an explicit ambiguity at the level of integrable boundary charges; it is not asserted to be representable by the restricted HHPS spacetime one-form $X_a$, or to preserve the same Hamiltonian action.

For $u_{AB}=\delta q_{AB}$ and $u=q^{AB}u_{AB}$, the two-dimensional Einstein tensor vanishes identically. The variation of curvature density consequently gives
\begin{equation}
 \delta F_f=\lambda\int_S\sqrt q
 \left[u_{AB}D^AD^Bf-uD^2f\right]\dd^2x.
\end{equation}
Boundary terms vanish only because $S$ is closed and the fields are smooth. For constant $f$, $F_f=4\pi\lambda f\,\chi(S)$ by Gauss--Bonnet, so $\delta F_f=0$. Such an improvement is invisible to the constant zero mode while changing nonconstant modes.

To evaluate it exactly on a conformal family, set $q_{AB}=e^{2s}\bar q_{AB}$ for a unit round sphere. Then
\begin{equation}
 \sqrt qR[q]=\sqrt{\bar q}(2-2\bar\Delta s),\qquad
 \delta F_f=-2\lambda\int_S\sqrt{\bar q}(\bar\Delta f)\delta s\dd^2x.
\end{equation}
For a real spherical harmonic normalized by $\int_S\sqrt{\bar q}\,Y_{\ell m}^2\dd^2x=1$, choose $f=Y_{\ell m}$ and $s=uY_{\ell m}$. Since $\bar\Delta Y_{\ell m}=-\ell(\ell+1)Y_{\ell m}$,
\begin{equation}
 \delta F_f=2\lambda\ell(\ell+1)\dd u.
\end{equation}
The quadrupole $\ell=2$ changes the charge by $12\lambda\dd u$ even though the constant smearing sees no change. This is a concrete homogeneous solution of the charge-closure equation. Requiring the original symmetry action with the same reduced symplectic form excludes nonconstant $F_f$ by the Hamiltonian uniqueness proposition. Thus action matching, not closure or zero-mode entropy alone, is the needed discriminator.

At the central-coefficient level a separate test also shows the limitation of temperature independence. Allow an algebraic extension of the ansatz,
\begin{align}
 K_R^X&=J\left[u\frac{T_L-T_R}{T_L+T_R}+v\right]m^3\delta_{m+n,0},\\
 K_L^X&=J\left[u\frac{T_R-T_L}{T_L+T_R}+v\right]m^3\delta_{m+n,0}.
\end{align}
The total coefficients are $J[u+v+2(1-u)x]$ and $J[2-u+v+2(u-1)x]$. Constancy selects $u=1$ but leaves $v$ unrestricted, with $c_L=c_R=12J(1+v)$. This is an algebraic underdetermination example. Existence of a local admissible counterterm producing $v\ne0$ has not been demonstrated; that existence question belongs to the local classification. The example explains why constancy by itself cannot replace the classification theorem.

\section{A gravitational corner polarization and its regulated spectrum}\label{cornerspectrum}
An area--boost canonical pair is grounded in gravitational corner potentials, rather than chosen only to fit Virasoro commutators. For a timelike stretched boundary meeting a spacelike slice, a standard relative corner polarization has potential \cite{fgp,hayward}
\begin{equation}\label{cornerpotential}
 \Theta_S=\frac1{8\pi G}\int_S a\,\delta\eta\dd^2x,
 \qquad \Omega_S=\delta\Theta_S
 =\frac1{8\pi G}\int_S\delta a\wedge\delta\eta\dd^2x.
\end{equation}
Here $\eta$ is a real Lorentzian boost angle. This is a relative corner contribution for a specified boundary formulation, not the entire gravitational symplectic form. Its retention and the null limit require an explicit choice of boundary data. The following calculation quantizes this component with a finite cell regulator; it does not solve bulk constraints or physical gluing.

Partition the corner into $N$ cells, retaining their positive areas $A_i$ and constant boost angles $\eta_i$. The regulated form is
\begin{equation}
 \Omega_N=\frac1{8\pi G}\sum_{i=1}^N\delta A_i\wedge\delta\eta_i.
\end{equation}
Choose a positive reference area $A_*$ and define
\begin{equation}
 q_i=\log(A_i/A_*),\qquad p_i=\frac{A_i\eta_i}{8\pi G}.
\end{equation}
Direct substitution gives the Darboux form
\begin{equation}
 \Omega_N=\sum_i\delta q_i\wedge\delta p_i.
\end{equation}
Quantize on $\HH_N=L^2(\mathbb R^N,\dd^Nq)$ with $\hat q_i$ acting by multiplication and $\hat p_i=-i\hbar\partial_{q_i}$. The positive cell-area operators are
\begin{equation}
 \hat A_i=A_*e^{q_i},\qquad \hat A_{\rm tot}=A_*\sum_i e^{q_i}.
\end{equation}
Their natural maximal multiplication domains give self-adjoint positive operators. The spectrum of each cell area and total area is $[0,\infty)$; zero is an accumulation point, not a normalizable eigenvalue. Positive area values also have no normalizable eigenvectors in this measure. Positivity of area does not imply a discrete area spectrum.

If, at fixed $\kappa$, one uses the corner Hamiltonian
\begin{equation}
 \hat H_\kappa=\frac{\kappa}{8\pi G}\hat A_{\rm tot},
\end{equation}
its spectrum is continuous and the Gibbs operator is not trace class. To see this without a formal $\delta(0)$, choose a bounded set $E$ of positive measure where $A_{\rm tot}\le A_0$. There exists an infinite orthonormal family $\psi_n$ supported in $E$. For $\beta>0$,
\begin{equation}
 \langle\psi_n,e^{-\beta\hat H_\kappa}\psi_n\rangle
 \ge e^{-\beta\kappa A_0/(8\pi G)}>0.
\end{equation}
The sum over this family diverges. Thus $\operatorname{Tr}e^{-\beta\hat H_\kappa}=\infty$. The fixed-area level set is likewise not a finite-dimensional microcanonical eigenspace.

This finite-regulator calculation determines a genuine kinematic spectrum for the chosen corner polarization. It shows that this component alone cannot provide a finite black hole state count. The cell regulator controls the number of canonical pairs, but does not compactify their phase space or fix their quantum multiplicities. Additional constraints, dynamics, or an independently motivated alternative representation are required. Imposing periodicity on the real boost angle merely to obtain discrete areas would be an added assumption: Euclidean angular periodicity is not automatically a compact Lorentzian canonical coordinate.

\section{Further representation and state-counting calculations}\label{quantchecks}
The Virasoro realization can be checked beyond individual descendant norms. At level two, the basis $(L_{-2}|h\rangle,L_{-1}^2|h\rangle)$ has Gram matrix
\begin{equation}
 G_2=\begin{pmatrix}4h+c/2&6h\\6h&4h(2h+1)\end{pmatrix},
\end{equation}
obtained solely from the commutators in \eqref{vir}. Its determinant is
\begin{equation}
 \det G_2=4h\left[8h^2+(c-5)h+\frac c2\right].
\end{equation}
For $c=24$, $h=1$, this is $156$ and the eigenvalues are $14\pm\sqrt{40}$, approximately $7.67544468$ and $20.32455532$. This verifies positivity at level two in an example. All-level positivity uses the representation theorem; a finite-level calculation does not replace it.

\subsection{Exact descendant multiplicities}
For a generic module with no null descendants, the generating product \eqref{character} gives exact partition numbers. The following values are calculated by integer dynamic programming:
\begin{center}
\begin{tabular}{rrrr}
\toprule
$N$ & $p(N)$ & $\log p(N)$ & $2\pi\sqrt{24N/6}$\\
\midrule
10 & 42 & 3.737670 & 39.738353\\
25 & 1958 & 7.579679 & 62.831853\\
50 & 204226 & 12.226983 & 88.857659\\
100 & 190569292 & 19.065526 & 125.663706\\
\bottomrule
\end{tabular}
\end{center}
The last column is the leading $c=24$ Cardy benchmark at cylinder energy $E=N$, not a prediction for a single module or an accurate finite-level result for an unspecified CFT. Choosing $h=1$ makes $E=N+h-c/24=N$ in this algebraic example. There is no gravitational identification of these example levels with $M$ and $J$.

For two copies at $N_L=N_R=100$, the module count is $p(100)^2$, with entropy $38.131053$. The two-sector leading Cardy benchmark is $251.327412$. Matching the leading exponent by adding equal-energy copies at that level would require an enhancement factor of approximately $\exp(213.196359)$. This factor is a diagnostic of missing spectral content, not a multiplicity prescription. Deriving primary modules and their energies is indispensable.

\subsection{A quantitative low-temperature obstruction}
Introduce the Kerr dimensionless spin $j=J/M^2\in(0,1)$ and $s=\sqrt{1-j^2}$. The conformal temperatures are
\begin{equation}
 T_L=\frac1{2\pi j},\qquad T_R=\frac{s}{2\pi j}.
\end{equation}
The conventional modular self-dual scale for a circle of circumference $2\pi$ is $T_*=1/(2\pi)$. While $T_L>T_*$ throughout the nonextremal interval, the right temperature satisfies
\begin{equation}
 T_R>T_*\quad\Longleftrightarrow\quad j<\frac1{\sqrt2}.
\end{equation}
For $j>1/\sqrt2$, the right sector lies below that scale; near extremality it approaches zero. This inequality is not by itself a proof of vacuum dominance even in the other region. It establishes that a uniform argument assuming both sectors are above the modular crossover fails for rapidly rotating holes.

\begin{center}
\begin{tabular}{rrr}
\toprule
$j$ & $T_L/T_*$ & $T_R/T_*$\\
\midrule
0.50 & 2.000000 & 1.732051\\
0.80 & 1.250000 & 0.750000\\
0.95 & 1.052632 & 0.328684\\
0.99 & 1.010101 & 0.142492\\
\bottomrule
\end{tabular}
\end{center}
A microscopic treatment in the mixed-temperature regime must analyze its actual left--right spectrum. Continuing the canonical Cardy expression into this region simply because it yields $A/4$ would use the entropy match as an assumption.

\section{Mixed parameter--mode variations and moving regulators}\label{mixedmodes}
The regularity calculation can be strengthened by varying the conformal temperature. Set $\alpha=2\pi T_R$ and use the same identification of the local $w=w^+$ chart across configurations. The mode components are
\begin{equation}
 \zeta_n^w=\alpha w e^{in\log w/\alpha},\qquad
 \zeta_n^y=\frac y2(\alpha+in)e^{in\log w/\alpha}.
\end{equation}
Their parameter derivatives at fixed $(w,y)$ are
\begin{align}\label{mixedderivatives}
 \partial_\alpha\zeta_n^w
 &=w e^{in\log w/\alpha}\left(1-\frac{in\log w}{\alpha}\right),\\
 \partial_\alpha\zeta_n^y
 &=\frac y2 e^{in\log w/\alpha}
 \left[1-\frac{in(\alpha+in)\log w}{\alpha^2}\right].
\end{align}
For $n\ne0$, the second expression grows logarithmically near $w=0$. A space that admits bounded oscillatory $\zeta_n^y$ need not admit these mixed variations. At fixed $J$, the temperature does vary with $M$:
\begin{equation}
 \alpha=\frac{\sqrt{M^4-J^2}}{J}=\frac DJ,
 \qquad \partial_M\alpha=\frac{2M^3}{JD}>0.
\end{equation}
Thus these derivatives are relevant on the fixed-$J$ family used in the central-charge analysis. Freezing $J$ does not freeze the conformal mode frame. Terms containing $\delta\zeta_n$ in a field-dependent charge prescription cannot be omitted on this family.

\subsection{A cancellation that must precede the horizon limit}
Consider the regulated kernel already encountered in Section \ref{regularity}:
\begin{equation}
 I_n(\alpha,\epsilon)=\int_\epsilon^{b(\alpha,\epsilon)}
 w^{in/\alpha}\frac{\dd w}{w},\qquad
 b(\alpha,\epsilon)=\epsilon e^{2\pi\alpha}.
\end{equation}
For $n\ne0$, $I_n=0$ exactly. Differentiating before taking any limit produces two nonzero pieces:
\begin{align}
 E_n&=\left.w^{in/\alpha}w^{-1}\partial_\alpha b\right|_{w=b}
 =2\pi\epsilon^{in/\alpha},\\
 B_n&=\int_\epsilon^b\left[-\frac{in}{\alpha^2}\log w\right]
 w^{in/\alpha}\frac{\dd w}{w}
 =-2\pi\epsilon^{in/\alpha}.
\end{align}
The equality for $B_n$ follows by setting $u=\log w$ and integrating $u e^{inu/\alpha}$ over an interval of length $2\pi\alpha$. Therefore $E_n+B_n=0$, although neither term has a limit as $\epsilon\to0$ for nonzero $n$. For $n=0$, the bulk derivative is zero and the endpoint contribution is $2\pi$, correctly differentiating $I_0=2\pi\alpha$.

This is an exact mixed-variation calculation for the toy pole kernel, not for the entire HHPS charge. It demonstrates why treating the cutoff surface as fixed while varying the metric-dependent azimuthal identification can produce a false obstruction or miss a real one. The complete gravitational calculation must retain both moving-surface and varying-generator terms before the limit.

Using $x=\log w/\alpha$ instead makes the modes $e^{inx}$ independent of $\alpha$, but changes the identification of spacetime points: $w=e^{\alpha x}$. The corner embedding and its variation then carry the missing dependence. A choice of field-space connection for the mode frame must be accompanied by the corresponding embedding terms. It is not a proof of integrability to set those terms to zero by notation.

\section{A weighted-space test of pole-trace continuity}\label{weightedtest}
A possible weak completion must accommodate oscillatory modes and their parameter derivatives. For a scalar component, take the weighted norm
\begin{equation}
 \|f\|_{H^1_p}^2=\int_0^1 w^p(|f|^2+|f'|^2)\dd w.
\end{equation}
For a nonzero Fourier frequency $k$, $f=w^{ik}$ has $|f'|=|k|/w$. Its norm is finite when $p>1$. A parameter derivative contributes $f\log w$; its derivative behaves as $(1+ik\log w)f/w$ and is also square integrable for $p>1$. Weighted spaces of this type therefore pass an elementary membership test for individual modes and mixed variations.

Membership is insufficient to control the regulated charge kernel. Let $C=e^{2\pi\alpha}$ and choose a nonnegative, nonzero $\varphi\in C_c^\infty((1,C))$. For small enough $\epsilon$, define $f_\epsilon(w)=\varphi(w/\epsilon)$. Substitution gives
\begin{equation}
 \|f_\epsilon\|_{H^1_p}^2=
 \epsilon^{p+1}\int_1^C u^p|\varphi(u)|^2\dd u
 +\epsilon^{p-1}\int_1^C u^p|\varphi'(u)|^2\dd u
 \longrightarrow0\quad(p>1).
\end{equation}
Nevertheless the regulated pole reading is independent of $\epsilon$:
\begin{equation}
 L_\epsilon(f_\epsilon)=\int_\epsilon^{C\epsilon}
 \frac{f_\epsilon(w)}{w}\dd w
 =\int_1^C\frac{\varphi(u)}u\dd u>0.
\end{equation}
Consequently the family $L_\epsilon$ has no uniform boundedness on this weighted space as $\epsilon\to0$. Its operator norms grow at least as $\epsilon^{-(p-1)/2}$.

This is a precise function-space diagnostic, not an assertion that these test functions are gravitational solution perturbations. The test shows that simply choosing $p>1$ to admit the HHPS poles does not supply a uniform estimate for all data in that space. A physical phase space might exclude the test family through constraints, or add a norm on the periodic boundary profiles. Either alternative must be stated and proved compatible with the field equations and nonlinear transformations. A finite central integral does not provide this missing estimate.

A useful mathematical target is a space equipped with both a weighted bulk norm and an explicit norm on the boundary Fourier data. Its temperature-dependent periodic profiles form a bundle over the stationary parameters. A connection identifies those profiles under mass variation. The previous section shows that the resulting embedding and connection terms must enter the charge calculation. We have identified a candidate class of estimates to seek, not constructed a closed nonlinear gravitational space with them.

\section{A local symplectic completion of a parameter-dependent mode plane}\label{mixedsymplectic}
The phase-space identification also has a mixed-direction consistency test. At fixed central charge, a Virasoro leaf may have a zero-mode label that varies with the stationary parameters. For a mode pair, its reference KKS pairing is proportional, in the standard shifted convention, to
\begin{equation}
 A_n(M)=2n h(M)+\frac c{12}n(n^2-1),\qquad n\ge1.
\end{equation}
The zero-mode label $h(M)$ is not derived from gravity here. Equivalently the cubic convention uses $b(M)=h(M)-c/24$ and $A_n=2nb+(c/12)n^3$. Real mode-plane coordinates can be normalized so that the reference pairing is $A_n\delta x\wedge\delta y$.

For clarity retain just one plane with $A(M)>0$. Adding a stationary canonical pair naively would give
\begin{equation}
 \Omega_{\rm naive}=\delta M\wedge\delta\tau+A(M)\delta x\wedge\delta y.
\end{equation}
Its exterior derivative is
\begin{equation}
 \delta\Omega_{\rm naive}=A'(M)\delta M\wedge\delta x\wedge\delta y,
\end{equation}
so it is not symplectic when $A'$ is nonzero. Merely taking a union of correctly normalized leaves does not define a symplectic phase space.

A local exact completion is obtained from
\begin{equation}
 \Theta=M\delta\tau+\tfrac12A(M)(x\delta y-y\delta x).
\end{equation}
Differentiation supplies the necessary mixed terms:
\begin{equation}\label{completedform}
 \Omega=\delta M\wedge\delta\tau+A\delta x\wedge\delta y
 +\tfrac12A'\delta M\wedge(x\delta y-y\delta x).
\end{equation}
It is closed by construction. With $X=\sqrt A\,x$ and $Y=\sqrt A\,y$, one verifies
\begin{equation}
 \delta X\wedge\delta Y=A\delta x\wedge\delta y
 +\tfrac12A'\delta M\wedge(x\delta y-y\delta x),
 \qquad \Omega=\delta M\wedge\delta\tau+\delta X\wedge\delta Y.
\end{equation}
This form is nondegenerate wherever $A>0$. Its square is $2A\delta M\wedge\delta\tau\wedge\delta x\wedge\delta y$.

The calculation supplies an explicit local solution to a mixed symplectic closure problem. It is an auxiliary finite-dimensional model of one tangent plane, not a finite-mode nonlinear Virasoro truncation: generic retained modes do not form a closed subalgebra. We have neither derived $\tau$ as an allowed gravitational edge variable nor calculated the gravitational coefficient $A(M)$ or the corresponding cross terms. The result establishes exactly what a phase-space map would have to reproduce, including mixed components, rather than permitting a match only on fixed-parameter leaves.

\section{Explicit modular spectra with identical central charges}\label{modularcounterexample}
There is also a constructive way to test whether quantum consistency and $c_L=c_R=12J$ select a spectrum. Compact free bosons provide unitary local CFTs with radius moduli and modular-invariant torus partition functions \cite{narain,mw}. We use them only as counterexamples to spectral selection by central charge and modularity. No identification with the Kerr gravitational phase space is asserted.

In units $\alpha'=1$, one compact boson of radius $R>0$ has
\begin{equation}
 p_L=\frac nR+wR,\quad p_R=\frac nR-wR,
 \quad h=\frac{p_L^2}{4}+N_L,\quad \bar h=\frac{p_R^2}{4}+N_R,
 \qquad n,w\in\mathbb Z.
\end{equation}
Each chirality has central charge one and $h-\bar h=nw+N_L-N_R$ is integral. For $d$ independent bosons, $c_L=c_R=d$ at every radius. Choosing $d=24$ gives the HHPS central value at $J=2$. Arbitrarily large integer $d$ similarly supplies arbitrarily large central-charge examples, so the underdetermination is not confined to a small value of $c$.

The one-boson Euclidean partition function is
\begin{equation}\label{compactZ}
 Z_R(\tau,\bar\tau)=\frac1{|\eta(\tau)|^2}
 \sum_{n,w\in\mathbb Z}
 e^{-\pi\tau_2(n^2/R^2+w^2R^2)+2\pi i\tau_1nw}.
\end{equation}
Integer spin proves invariance under $\tau\mapsto\tau+1$. Poisson resummation in $n$ gives the equivalent expression
\begin{equation}
 Z_R=\frac R{\sqrt{\tau_2}|\eta(\tau)|^2}
 \sum_{m,w\in\mathbb Z}
 \exp\left[-\frac{\pi R^2}{\tau_2}|m-w\tau|^2\right].
\end{equation}
Under $\tau\mapsto-1/\tau$, relabeling $(m,w)$ and using
$|\eta(-1/\tau)|^2=|\tau||\eta(\tau)|^2$ proves modular invariance. The result holds for every positive radius. Radius inversion exchanges momentum and winding, but radii not related by that duality have distinct spectra.

The lowest nonzero total primary dimension is
\begin{equation}
 \Delta_{\min}(R)=\frac12\min(R^2,R^{-2}).
\end{equation}
It is $1/2$ at $R=1$, $1/8$ at $R=2$, and $1/50$ at $R=5$, while $c$ is identical. In particular a change of radius changes the gap and the finite-temperature state density without changing the Virasoro central terms or modular consistency.

\begin{proposition}\label{spectralnonselection}
Central charges $c_L=c_R=12J$, unitarity, locality, integer spin, and torus modular invariance do not in general determine a unique physical state spectrum.
\end{proposition}
\begin{proof}
At $12J=d\in\mathbb Z_{>0}$, tensor products of $d$ compact bosons satisfy all the stated quantum consistency properties at every radius. The distinct primary gaps above distinguish the spectra. None of those properties selects a radius or the required gravitational map.
\end{proof}
The counterexample also persists at large central charge. Take $d=24K$, $J=2K$, $M=2\sqrt K$, and the Kerr rotation parameter $a=\sqrt K$ for positive integer $K$. The spin ratio remains $1/2$, so the conformal temperatures are independent of $K$. Each candidate CFT entropy below and the Kerr area entropy scale by $K$. Their differences are therefore extensive in $c$, not merely negligible finite-$c$ corrections. These tensor-product theories do not satisfy the required sparse-spectrum assumptions for a universal replacement of their trace at the stated temperatures.

The examples have additional current symmetries. The proposition concerns the explicitly stated conditions, not every possible condition a gravitational completion might impose. Prohibiting these currents would be an additional selection input whose origin would have to be established. The examples do not prove that Kerr has multiple microscopic theories; they prove that the present symmetry data alone do not specify one.

\subsection{Exact thermal traces at the Kerr conformal temperatures}
For independent positive inverse temperatures $\beta_L,\beta_R$, the one-boson trace is an analytic continuation of the torus trace and is absolutely convergent:
\begin{equation}
 Z_R(\beta_L,\beta_R)=
 e^{(\beta_L+\beta_R)/24}
 P(\beta_L)P(\beta_R)\mathcal T_R(\beta_L,\beta_R),
\end{equation}
where
\begin{equation}
 P(\beta)=\prod_{k=1}^{\infty}(1-e^{-\beta k})^{-1},\qquad
 \mathcal T_R=\sum_{n,w}\exp\left[-\frac{\beta_Lp_L^2+\beta_Rp_R^2}{4}\right].
\end{equation}
The entropy operator is $1-\beta_L\partial_{\beta_L}-\beta_R\partial_{\beta_R}$. The linear vacuum shift cancels. Thus
\begin{align}
 S_R&=S_{\rm osc}(\beta_L)+S_{\rm osc}(\beta_R)
 +\log\mathcal T_R+\beta_L\langle p_L^2/4\rangle
 +\beta_R\langle p_R^2/4\rangle,\\
 S_{\rm osc}(\beta)&=\sum_{k=1}^{\infty}
 \left[-\log(1-e^{-\beta k})+\frac{\beta k}{e^{\beta k}-1}\right].
\end{align}
For $d$ equal-radius independent bosons the entropy is $dS_R$. This formula counts the actual spectrum of the candidate CFT rather than assuming Cardy vacuum dominance.

At the illustrative Kerr values $M=2$, $J=2$, $a=1$, take $\beta_L=\pi$ and $\beta_R=2\pi/\sqrt3$. For $d=24$ we obtain:
\begin{center}
\begin{tabular}{rrrr}
\toprule
$R$ & $\Delta_{\min}$ & $\mathcal T_R$ & $S_{\rm exact\ trace}$\\
\midrule
1 & 0.5000 & 1.8824710611 & 45.57087210\\
2 & 0.1250 & 2.7314824410 & 44.48329881\\
5 & 0.0200 & 6.8125003863 & 66.03369254\\
\bottomrule
\end{tabular}
\end{center}
The conditional leading Cardy value, and $A/4$ for this Kerr example, is $46.89833360$. The table evaluates convergent traces numerically, not exact decimal identities. At a common central charge and common temperatures, candidate microscopic entropies differ. Taking $R=5$ does not make an invalid CFT; it makes a CFT whose light spectrum does not justify replacing its trace by the desired leading Cardy expression at these temperatures.

For numerical control, truncate $|n|,|w|\le K$. If $\beta_{\min}=\min(\beta_L,\beta_R)$, set $a_0=\beta_{\min}/(2R^2)$ and $b_0=\beta_{\min}R^2/2$. The exponent obeys $E\ge a_0n^2+b_0w^2$. Therefore the omitted lattice partition sum is bounded by
\begin{equation}
 \delta\mathcal T_R\le t(a_0,K)s(b_0)+s(a_0)t(b_0,K),
\end{equation}
where
\begin{equation}
 t(a,K)=\sqrt{\frac\pi a}\operatorname{erfc}(\sqrt aK),\qquad
 s(a)=1+\sqrt{\frac\pi a}.
\end{equation}
These bounds follow by comparison of decreasing Gaussian sums with integrals. At $K=40$ the worst partition-sum bound among the tabulated radii is below $2\times10^{-44}$. The bound is for the partition sum, not directly for its energy-weighted derivatives. Entropies agree under increasing $K$ from 40 to 60 at the reported precision; the oscillator sum is truncated at 99 with exponentially negligible tails for these temperatures. Floating-point arithmetic limits the last digits.

\section{Finite entropy and infinite-dimensional symmetry}\label{finitetrace}
One might try to declare the black hole Hilbert space finite dimensional with dimension $e^{A/4}$ and then impose the full Virasoro algebra on it. There is an elementary obstruction to that construction.
\begin{proposition}
A nonzero finite-dimensional representation of the full Virasoro algebra with central element $c\mathbf1$ cannot have $c\ne0$.
\end{proposition}
\begin{proof}
Let its dimension be $D_H>0$. Taking traces of $[L_1,L_{-1}]=2L_0$ yields $\operatorname{Tr}L_0=0$. Taking traces of $[L_2,L_{-2}]=4L_0+c\mathbf1/2$ then yields $cD_H/2=0$. Hence $c=0$.
\end{proof}
This does not obstruct a finite microcanonical entropy inside an infinite-dimensional theory. A fixed-energy subspace can have finite degeneracy while $L_n$ maps it to other energy sectors; that subspace is not a representation of the full algebra. A genuine microstate calculation must specify the macroscopic charge window and count its physical states within the larger representation. Assigning a finite dimension to the entire Virasoro Hilbert space would conflate these two objects.

\section{The strengthened completion status}\label{strongstatus}
The mixed-variation calculation identifies logarithmic temperature dependence, and the moving-regulator identity exhibits a cancellation that must be retained before the horizon limit. The weighted-space example gives a concrete estimate that a weak phase-space construction must overcome. The mode-plane calculation supplies a local symplectic completion with explicit cross terms, while making clear which coefficients and canonical variables remain unproved on the gravitational side.

The compact-boson examples then show that even a well-defined unitary local modular quantum theory with the required central charges need not have the desired state spectrum or finite-temperature entropy. These results move the analysis beyond abstract caveats: several tempting routes to claiming completion now have calculable failures or missing identifications.

We have not established a complete Kerr microstate theory. In particular, we have not derived a nonlinear regulated HHPS gravitational symplectic form, its global reduction and gluing, or a unique physical spectrum. The new results do not prove that such a theory is impossible. They prove that the supplied classical central data and the listed quantum consistency properties are insufficient to select it. A microscopic completion would need independent gravitational or quantum dynamics that remove this underdetermination. Choosing a radius, orbit label, measure, or degeneracy only because it gives the area entropy would not supply those dynamics.

\section{What the additional calculations establish}\label{newstatus}
The results supply concrete progress beyond a list of desired assumptions. On the exact Kerr family, an explicit nonzero curl identifies the normalization and phase-space dependence of the horizon zero-mode charge. A smooth fixed-temperature boundary sector has nonlinear integrable charges, but its directly computed algebra has no nontrivial central extension. Explicit derivatives of the HHPS modes show why those sectors cannot be identified without changing the regularity class.

A curvature improvement gives an explicit local boundary functional invisible to the constant mode but visible to higher harmonics. The area--boost corner polarization supplies Darboux coordinates and a regulated kinematic quantization with continuous spectrum and a divergent Gibbs trace. Exact Virasoro descendant counts quantify why algebraic representations alone do not determine the desired large-$c$ state density. The spin threshold exposes the mixed-temperature regime that a full statistical argument must handle.

These are calculations within clearly stated sectors and models. The integrable smooth charges and the corner polarization have gravitational origins, but their combination is not asserted to be the full reduced gravitational phase space. Their failure to supply the HHPS central extension or a finite entropy is informative: an identification with the proposed microstate theory requires additional boundary structure and derived spectral content. The calculations of Sections \ref{mixedmodes}--\ref{strongstatus} further isolate mixed-direction and spectral-selection failures. The next unresolved calculation is the nonlinear gravitational symplectic form in a precisely defined regulated HHPS mode space, including the regulator embedding, mixed variations and the associated physical constraints.

\section{A conditional completion theorem}
\begin{theorem}[Sufficiency of independently established completion data]
Fix a connected nonextremal rotating sector at $J>0$. Suppose the following are established from a gravitational theory and specified boundary conditions:
\begin{enumerate}
\item A reduced phase space with a closed nondegenerate symplectic form, preserved by the adjusted left and right symmetry actions.
\item Finite charge one-forms with zero field-space curls and zero periods, including control of regulator limits and all allowed transverse perturbations.
\item A counterterm classification whose solutions yield one normalized Hamiltonian action modulo charge-invisible terms, and the gravitational evaluation of its cubic coefficients as $J$.
\item A quantization of these gravitational charges with a positive physical inner product, appropriate constraint gluing, and quantum central charges $12J$; possible quantum shifts are absent or independently accounted for.
\item A derived primary spectrum and multiplicity measure admitting a statistical ensemble in which the leading formula \eqref{cardy} is justified at the physical temperatures.
\end{enumerate}
Then the sector has globally integrable gravitational symmetry charges, a physical quantum realization of their algebra, and leading semiclassical entropy $A/4$. Counterterms are unique in the stated physical equivalence class, not necessarily as spacetime one-forms.
\end{theorem}
\begin{proof}
Items 1 and 2 give normalized global Hamiltonians by Proposition \ref{integrability}. Item 3 fixes their central values and the equivalence class of the boundary prescription. Item 4 supplies the required quantum realization and physical constraints. Item 5 justifies the leading statistical formula; the Kerr identity in Section 7 then gives $A/4$.
\end{proof}
The theorem organizes necessary independent inputs into a sufficient package; it does not establish its premises for HHPS. In particular, it must not be cited as a completed first-principles derivation. Failure of one premise identifies a specific remaining task rather than invalidating the arithmetic of the central term.

\section{A concrete path to a gravitational completion}
The first deliverable should be a nonlinear definition of the boundary phase space. One must list the permitted expansion coefficients near a stretched horizon, including the singular pure-diffeomorphism modes, and choose function spaces in which the field equations, charge integrals, and limits are controlled. The mass dependence of the mode frame must be fixed or treated with adjusted brackets. The residual null-frame boost must be fixed or represented as an edge degree of freedom.

The second deliverable should evaluate \eqref{curl}, rather than only a central coefficient. A useful working basis has mode labels and regular perturbations, but an analytic completeness or density argument is needed to extend tests from that basis to the allowed tangent space. Symbolic tensor calculations can assist the evaluation; finite mode checks can find counterexamples but cannot prove closure of an infinite-dimensional phase space. Estimates should record possible uncanceled poles, logarithms, and cutoff dependence.

The third deliverable is the counterterm classification in a stated local derivative class. Its equations combine \eqref{cancel}, frame covariance, finiteness, and Hamiltonian normalization. The resulting homogeneous solution space must be mapped to charge improvements and invisible terms. If nontrivial improvements survive, the theory needs independent boundary physics to choose among them; matching the known entropy is not such an independent choice.

The fourth deliverable is a gravitational quantization or a demonstrated symplectic identification with a quantizable model. For the orbit route, the map \eqref{map} and the gravitational zero-mode values must be computed. The central term alone supplies neither the orbit labels nor the polarization, measure, or possible quantum shifts. The finite descendant domain should be distinguished from global self-adjoint observables.

The final deliverable is the derived spectrum and gluing map. State counting must retain physical multiplicities, discard gauge copies, and justify its ensemble. If modularity is proposed, one must construct the partition function and test modular covariance with the derived spectrum. If a different counting mechanism is used, it must calculate the density independently and explain why it reproduces the area. Near-extremal and nonrotating limits should be treated separately rather than assumed to follow from the same conformal chart.

\begin{center}
\small
\begin{tabular}{p{0.25\linewidth}p{0.35\linewidth}p{0.29\linewidth}}
\toprule
Question & Result supplied here & Remaining gravitational task\\
\midrule
Central coefficient & $\lambda=1$ uniquely in scalar-multiple ansatz & Independent full integral and boundary analysis\\
Charge integrability & Kerr curl; exact smooth-sector primitive & Verify regulated HHPS mixed variations\\
Counterterm uniqueness & Explicit curvature improvement & Classify HHPS local representatives and kernel\\
Quantum realization & Virasoro modules; regulated corner spectrum & Quantize full reduced space and impose constraints\\
Entropy & Exact descendant counts; divergent corner trace & Derive physical multiplicities and valid ensemble\\
\bottomrule
\end{tabular}
\end{center}

\section{Conclusion}
The HHPS calculation correctly combines the stated Iyer--Wald and Wald--Zoupas central terms to obtain $c_L=c_R=12J$. Temperature independence also fixes the counterterm coefficient within its scalar-multiple ansatz. These are meaningful results, but they do not supply the full nonlinear gravitational theory.

Integrability requires closure on all allowed tangent pairs and vanishing global periods. Counterterm uniqueness must be formulated relative to an admissible local class and a quotient by terms that do not change the normalized Hamiltonian action. Explicit Virasoro modules show that quantum algebraic realization is possible for the macroscopic central values, while their characters show why that realization does not determine the black hole density of states. Primary multiplicities, physical gluing, and the statistical regime are essential additional data.

The additional calculations establish a Kerr zero-mode curl, an integrable smooth boundary sector whose algebra has zero central extension, a regularity obstacle for the nonzero HHPS modes, and a charge improvement invisible to the constant smearing. The chosen area--boost corner polarization has a continuous regulated kinematic spectrum and no trace-class Gibbs operator. The mixed parameter--mode and weighted-space calculations specify additional regulator requirements, while distinct modular spectra at fixed central charge establish spectral nonselection, including at large $c$. These results identify concrete failures of simple completion attempts and sharpen the gravitational and spectral tasks that remain.

The completion problem is therefore reduced to concrete computations rather than resolved by an unsupported assertion. Establishing the gravitational phase space, solving the full obstruction equation, classifying its homogeneous solutions, and deriving the physical spectrum would turn the conditional entropy match into a first-principles microstate calculation. None of these missing premises is silently assumed to have been proved in this paper.

\appendix
\section{Reproducible numerical checks}
The following Python 3 code uses only the standard library. It checks the two Kerr first-law derivative identities and the curl coefficient at three nonextremal backgrounds, reproduces the finite-loop table and exact integer partition counts, and prints the level-two Gram eigenvalues and temperature ratios. Finite differences verify the analytic formulas numerically; the analytic derivations remain the basis of the results. No numerical test here purports to verify the nonlinear HHPS symplectic form, a global spectrum, or a quantum gravity state count.

\lstset{language=Python,basicstyle=\ttfamily\footnotesize,
 breaklines=true,columns=fullflexible,keepspaces=true,
 showstringspaces=false,frame=single}
\begin{lstlisting}
import math

def omega(m, j):
    d = math.sqrt(m**4 - j*j)
    return j / (2*m*(m*m + d))

def area(m, j):
    return 8*math.pi*(m*m + math.sqrt(m**4-j*j))

def kappa(m, j):
    d = math.sqrt(m**4-j*j)
    return d / (2*m*(m*m+d))

def curl(m, j):
    d = math.sqrt(m**4-j*j)
    return omega(m, j)*(1/m + 2*m/d)

def derivative(f, x, step=1e-5):
    return (f(x+step)-f(x-step))/(2*step)

def simpson(f, left, right, n=4000):
    step = (right-left)/n
    total = f(left)+f(right)
    for i in range(1, n):
        total += (4 if i % 2 else 2)*f(left+i*step)
    return step*total/3

for m, j in ((2., 2.), (3., 4.), (2., 3.5)):
    dm = derivative(lambda x: area(x, j), m)
    dj = derivative(lambda x: area(m, x), j)
    assert abs(kappa(m,j)*dm/(8*math.pi)-1) < 1e-7
    assert abs(kappa(m,j)*dj/(8*math.pi)+omega(m,j)) < 1e-7
    numeric = -derivative(lambda x: omega(x,j), m)
    assert abs(numeric-curl(m,j)) < 1e-7

print('curl at (2,2):', curl(2,2))
for e in (0.1, 0.05, 0.025, 0.0125):
    loop = simpson(lambda j: omega(2,j)-omega(2+e,j),
                   2, 2+e)
    print('loop:', e, loop, loop/(e*e))

p = [0]*101
p[0] = 1
for part in range(1, 101):
    for level in range(part, 101):
        p[level] += p[level-part]
assert p[10] == 42 and p[100] == 190569292
for n in (10, 25, 50, 100):
    benchmark = 2*math.pi*math.sqrt(24*n/6)
    print('states:', n, p[n], math.log(p[n]), benchmark)

print('Gram determinant:', 16*12-6*6)
print('Gram eigenvalues:', 14-math.sqrt(40),
      14+math.sqrt(40))
for j in (0.5, 0.8, 0.95, 0.99):
    print('temperature ratios:', j, 1/j,
          math.sqrt(1-j*j)/j)
print('two-module log count:', 2*math.log(p[100]))
print('two-sector benchmark:', 2*2*math.pi*math.sqrt(400))
\end{lstlisting}

\section{Checks of mixed derivatives and candidate thermal spectra}
This independent standard-library script verifies the analytic mixed mode derivative by a finite difference and reproduces the compact-boson entropy table. Increasing the lattice cutoff from 40 to 60 checks stability of the energy-weighted sums. The printed Gaussian tail estimate bounds the partition sum; a printed zero can result from floating-point underflow and does not mean the analytic tail vanishes. These calculations test explicitly specified candidate theories and do not identify them with Kerr microstates.

\begin{lstlisting}
import math
import cmath

def zy(alpha, n, w):
    phase = cmath.exp(1j*n*math.log(w)/alpha)
    return (alpha+1j*n)*phase/2

alpha, n, w = 0.7, 2, 1e-4
step = 1e-6
numeric = (zy(alpha+step,n,w)-zy(alpha-step,n,w))/(2*step)
phase = cmath.exp(1j*n*math.log(w)/alpha)
exact = phase*(1-1j*n*(alpha+1j*n)*math.log(w)/alpha**2)/2
assert abs(numeric-exact) < 1e-6
print('mixed derivative error:', abs(numeric-exact))

beta_l = math.pi
beta_r = 2*math.pi/math.sqrt(3)

def oscillator(beta):
    return sum(-math.log1p(-math.exp(-beta*k))
               + beta*k/math.expm1(beta*k)
               for k in range(1,100))

def lattice(radius, cutoff):
    z = e_l = e_r = 0.0
    for n in range(-cutoff,cutoff+1):
        for winding in range(-cutoff,cutoff+1):
            h_l = (n/radius+winding*radius)**2/4
            h_r = (n/radius-winding*radius)**2/4
            weight = math.exp(-beta_l*h_l-beta_r*h_r)
            z += weight
            e_l += weight*h_l
            e_r += weight*h_r
    entropy = math.log(z)+beta_l*e_l/z+beta_r*e_r/z
    return z, entropy

def partition_tail_bound(radius, cutoff):
    a = min(beta_l,beta_r)/(2*radius*radius)
    b = min(beta_l,beta_r)*radius*radius/2
    def tail(x):
        return math.sqrt(math.pi/x)*math.erfc(math.sqrt(x)*cutoff)
    def total(x):
        return 1+math.sqrt(math.pi/x)
    return tail(a)*total(b)+total(a)*tail(b)

s_osc = oscillator(beta_l)+oscillator(beta_r)
for radius in (1.,2.,5.):
    z, s = lattice(radius,40)
    z_more, s_more = lattice(radius,60)
    assert abs(s-s_more) < 1e-10
    print('radius, theta, entropy:', radius,z,24*(s+s_osc))
    print('partition tail bound:', partition_tail_bound(radius,40))
print('conditional Cardy value:',
      math.pi**2*24*(1/beta_l+1/beta_r)/3)
\end{lstlisting}

\end{document}